\documentclass[a4paper,english,12pt]{article}

\usepackage{babel}
\usepackage{latexsym}
\usepackage{amsmath}
\usepackage{amsfonts}
\usepackage{amsthm}

\usepackage{graphicx}                 
\usepackage{hyperref}

\usepackage[authoryear]{natbib}
 
\newtheorem{theorem}{Theorem}
\newtheorem{lemma}{Lemma}

\newtheorem{proposition}{Proposition}

\newtheorem{assumption}{Assumption}

\newcommand{\nc}{\newcommand}
\nc{\vekk}[1]{}

\nc{\lik}{ = }  \nc{\supp}{\mbox{supp}} \nc{\E}{\mbox{E}}

\nc{\ttau}{\tilde{\tau}} \nc{\ttheta}{\tilde{\theta}}

\nc{\fra}{From} \nc{\st}{|}
\newcommand{\bT}{{\mathbf{T}}}

\newcommand{\bt}{{\mathbf{t}}}

\newcommand{\bU}{{\mathbf{U}}}
\newcommand{\bu}{{\mathbf{u}}}
\newcommand{\bv}{{\mathbf{v}}}
\newcommand{\bX}{{\mathbf{X}}}
\newcommand{\bx}{{\mathbf{x}}}

\newcommand{\cU}{{\mathcal{U}}}
\newcommand{\cX}{{\mathcal{X}}}

\newcommand{\RealN}{{\mbox{${\mathbb R}$}}} 

\newcommand{\btheta}{\mbox{\boldmath$\theta$}}

\nc{\beq}{\begin{equation}}

\nc{\eeq}{\end{equation}}

\nc{\K}{K}

\nc{\beqns}{\begin{eqnarray*}}
\nc{\eeqns}{\end{eqnarray*}}
\nc{\beqn}{\begin{eqnarray}}
\nc{\eeqn}{\end{eqnarray}}

\nc{\hatt}{\hat{\theta}}

\nc{\beit}{\begin{itemize}}
\nc{\eit}{\end{itemize}}

\begin{document}

\baselineskip=1.5pc


\title{Conditional Monte Carlo revisited}

\author{Bo H. Lindqvist\thanks{bo.lindqvist@ntnu.no}, Rasmus Erlemann\thanks{rasmus.erlemann@ntnu.no}, Gunnar Taraldsen\thanks{gunnar.taraldsen@ntnu.no} \\ Department
of Mathematical Sciences \\ Norwegian University of Science and Technology \\ Trondheim, Norway}

  \date{}

\maketitle

\begin{abstract}
\noindent Conditional Monte Carlo refers to  sampling from the conditional distribution of a random vector $\bX$ given the value $T(\bX) =\bt$ for a 
function $T(\bX)$.  Classical conditional Monte Carlo methods were designed for estimating  conditional expectations of functions $\phi(\bX)$
by sampling from unconditional distributions obtained by certain weighting schemes. The basic ingredients were the use of importance sampling and change of variables. In the present paper we reformulate the problem by introducing an artificial parametric model, representing the conditional distribution of $\bX$ given $T(\bX)=\bt$ within this new model. The key is to provide the parameter of the artificial model by a distribution. The  approach is illustrated by several examples, which are particularly chosen to illustrate conditional sampling in cases where such sampling is not straightforward. A simulation study and an application to goodness-of-fit testing of real data are also given. 
\end{abstract}

\bigskip

\noindent%
\vfill

\newpage

\section{Introduction}

 Suppose we want to sample from the conditional distribution of a random vector  $\bX$ given $T(\bX)=\bt$ for a function $T(\bX)$ of $\bX$.
\cite{tukey} presented an interesting technique which they 
named \textit{conditional Monte Carlo.} Their idea was to
determine a weight $w_\bt (\bX)$ and a modified sample $\bX_\bt$ such that
$\E[\phi(\bX) \st T(\bX) \lik \bt]  \lik  \E [\phi(\bX_\bt) w_\bt (\bX)]$ 
for any function $\phi$, 
thus replacing conditional expectations by ordinary expectations
and allowing Monte Carlo computation. 

 Although the authors were aware that the
method had generalizations, they confined themselves to rather special cases. \cite{hammersley} used their idea in a slightly more general and flexible analytic setting, see also Chapter 6 of the monograph by \cite{hambook}. \cite{wendel} gave an alternative explanation, wherein the group-theoretic aspect of the problem played the dominant role. 
Later, \cite{dubi} gave an explanation of conditional Monte Carlo in terms of importance sampling and change of variables. Their approach provides a  framework by which in principle any
conditional sampling problem can be handled, 
and is the survivor in textbooks \citep{RIPLEY,EVANS}. 
Conditional Monte Carlo, in the form as introduced in the 1950s and the following nearest decades, has apparently received
little attention in the later literature and has seemingly remained theoretically underdeveloped. An interesting recent reference is \cite{feng} who exploit the change of variables framework of conditional Monte Carlo with application to  sensitivity estimation for financial options.

Sampling from conditional distributions has been of particular interest in statistical inference problems involving sufficient statistics \citep{lehmann,lehmanncasella}. Typical applications are in construction of optimal estimators, nuisance parameter elimination and goodness-of-fit testing. 
In some special cases one is able to derive conditional distributions analytically. Typically this is not possible, however, thus 
leading to the need for approximations or simulation algorithms.

\cite{engenlil} considered the general problem of Monte Carlo computation of conditional expectations given a 
sufficient statistic. Their approach was further studied and generalized by \cite{LT05}, see also  \cite{counter} and \cite{doksum}. Further applications of the technique can be found in \cite{schweder}, pp. 239, 250. Consider a statistical model where a random vector $\bX$ has a distribution indexed by the parameter $\theta$, and suppose the statistic $\bT$ is sufficient for $\theta$.
\vekk{ 
 The basic assumption is that there is given a random vector $\bU$  with a known distribution, such that $(\bX,\bT)$ for a
given parameter value $\theta$, say, can be simulated by means of $\bU$. More precisely, it is assumed that there exist functions $\chi$ and $\tau$ such 
that, for each $\theta$, the joint distribution of $(\chi (\bU, \theta), \tau (\bU, \theta))$  equals the joint distribution of 
$(\bX,\bT)$ under~$\theta$. 
} 
 The basic assumption is that there is given a random vector $\bU$  with a known distribution, such that $(\bX,\bT)$ for a
given parameter value $\theta$, say, can be simulated by $(\chi (\bU, \theta), \tau (\bU, \theta))$ for given  functions $\chi$ and $\tau$.  
Let $\bt$ be the observed value of $\bT$, and suppose that a sample from the conditional distribution of $\bX$ given $\bT \lik \bt$ is 
wanted. Since the conditional distribution by sufficiency does not depend on $\theta$, it seems reasonable that it can be described in some simple way in terms of the distribution of $\bU$, and thus enabling Monte Carlo simulation based on 
$\bU$. The main idea of \cite{engenlil} was to first draw $\bU=\bu$ from its known distribution, then to determine a parameter 
value $\hatt$ such that $\tau(\bu,\hatt)=t$ and finally to use $\chi(\bu,\hatt)$ as the desired sample. In this way one 
indeed gets a sample of $\bX$ with the corresponding $\bT$ having the correct value $\bt$. However, as shown by \cite{LT05}, only under a so-called pivotal condition will this be a sample from the true conditional distribution.  The clue \citep{LT05} is to let the parameter $\theta$ be given a suitable distribution, changing it to a random vector $\Theta$, independent of $\bU$, and showing that the conditional distribution of $\bX$ given $\bT=\bt$ equals the conditional distribution of $\chi(\bU,\Theta)$ given $\tau(\bU,\Theta) = \bt$.

In the present paper, motivated by the classical approaches of conditional Monte Carlo, we construct a method for sampling from conditional distributions of $\bX$ given $\bT \equiv T(\bX)= \bt$  in general, without reference to a particular statistical model and sufficiency. As was suggested in \cite{LT05}, this could in principle be done by embedding the pair $(\bX,\bT)$ in an artificial parametric model where $\bT$ is a sufficient statistic, and proceed as above. This may, however, often not be a simple task, if practically doable at all. While the new method is also based on the construction of an artificial parametric statistical model, sufficiency of $\bT$ is not part of this construction.  As will be demonstrated in examples, algorithms derived from the present approach will often be more attractive than the ones based on the sufficency approach as described above.  

The main idea of the new method is to construct an artificial statistical model for a random vector $\bU$ with distribution depending on a parameter $\theta$, such that a ``pivot'' $\chi(\bU,\theta)$ has the same distribution as $\bX$ for each $\theta$. Moreover, defining $\tau(\bU,\theta)=T(\chi(\bU,\theta))$, and  considering $\theta$ as the realization of a random $\Theta$, it will follow that the pair $(\chi(U,\Theta),\tau(U,\Theta))$ has the same distribution as $(\bX,\bT)$. Consequently, the conditional distribution of $\bX$ given $\bT=\bt$ equals the conditional distribution of $\chi(\bU,\Theta)$ given $\tau(\bU,\Theta) = \bt$. This end result similar to what was described above for the approach of \cite{LT05}, but a crucial difference from the latter approach  is that the $\bU$ and $\Theta$ are no longer independent.

As indicated above, an advantage of the new approach is that it applies to a single distribution for $\bX$ instead of a parametric model. Thus, when applied to conditional sampling given a sufficient statistic, the method may be based on the original model only under a conveniently chosen single parameter value, for example using a standard exponential distribution when the model is a two-parameter gamma model as in Section~\ref{421}. 

We give several examples to demonstrate the approach and illustrate different aspects of the theoretical derivations. In particular, the examples include a new method for sampling of uniformly distributed random variables conditional on their sum, where the method of embedding the distribution into a parametric family and using sufficiency is much less attractive than the new method. Other examples consider conditional sampling given sufficent statistics in the gamma and inverse Gaussian models, as well as a new treatment of a classical example from \cite{tukey}. 

The recent literature contains several other approaches to conditional sampling. For example, \cite{lockhartgibbs} and \cite{lockhartvonmises} studied the use of Gibbs sampling to generate samples from the conditional distribution given the minimal sufficient statistic for the gamma distribution and the von Mises distribution, respectively. \cite{medor} and \cite{ormed} constructed corresponding sampling methods based on the  Rao-Blackwell theorem, while \cite{santos} suggested a method using the Metropolis-Hastings algorithm. An older reference for conditional sampling in the inverse Gaussian distribution is \cite{chengIG}.                                             

The present paper is structured as follows. In Section~\ref{sec2} we explain the main method and prove the basic results underpinning the approach. Specific methods for simulation and computation within the approach are also briefly described. Some further explanations and theoretical extensions are given in Section~\ref{sec3}. Section~\ref{sec4} is devoted to examples, in particular for a general two-parameter exponential family of positive variables. Some simulation results which indicate the correctness of the methods are given in Section~\ref{sec5}, while an example of goodness-of-fit testing with real data is given in Section~\ref{sec6}. Some final remarks are given in Section~\ref{sec7}. The paper is concluded by an Appendix containing two lemmas referred to earlier in the paper.

\section{The main method}
\label{sec2}
Let $\bX$ be a random vector  and let $\bT=T(\bX)$ be a function of $\bX$. Our aim is to sample from the conditional distribution of $\bX$ given $\bT=\bt$. As indicated in the Introduction, the idea is to construct a pair $(\bU,\Theta)$ of random vectors and functions $\chi(\bU,\Theta)$ and $\tau(\bU,\Theta)$ such that  this conditional distribution equals the one of $\chi(\bU,\Theta)$ given $\tau(\bU,\Theta)=t$. 

Let $\bU$ be a random vector with values in $\cU$ and  distribution $P_\theta$ depending on a parameter $\theta \in \Omega$. Assume that there is a function $\chi(\bu,\theta)$ defined for $\bu \in \cU$, $\theta \in \Omega$, such that
\beq
\label{pivot}
\chi(\bU,\theta) \sim \bX \mbox{ for each $\theta \in \Omega$}.
\eeq
  Here '$\sim$' means 'having the same distribution as', and $\bU$ in (\ref{pivot}) is assumed to have the distibution $P_\theta$.      
Note that  $\chi(\bU,\theta)$ is then a \textit{pivot} in the statistical model defined by $\bU$ and $P_\theta$.  

The following result is basic to our approach. Let notation and assumptions be as above and let $\tau(\bu,\theta)= T(\chi(\bu,\theta))$ for $\bu \in \cU$ and $\theta \in \Omega$.
\begin{theorem}
\label{newth}
Let $\Theta$ be a random vector taking values in $\Omega$ and let $\bU$ conditional on $\Theta=\theta$ be distributed as $P_\theta$. If $\chi$ satisfies (\ref{pivot}), then  the conditional distribution of $\bX$ given $\bT=t$ is equal to the conditional distribution of $\chi(\bU,\Theta)$ given $\tau(\bU,\Theta)=t$.   
  \end{theorem}
 \begin{proof} It is enough to prove that $\chi(\bU,\Theta) \sim \bX$. Then it will follow that $(\chi(\bU,\Theta),\tau(\bU,\Theta)) \sim (\bX,\bT)$ and the result of the theorem will follow.   Now, by (\ref{pivot}), for any bounded function $\phi$,
 \[
 \E[\phi(\chi(\bU,\Theta))] = \E \left[ \E [\phi(\chi(\bU,\Theta))|\Theta \right]
 = \E[\phi(\bX)].
 \]
 Since this holds for all $\phi$, we conclude that $\chi(\bU,\Theta) \sim \bX$.
  \end{proof}  
The following result shows how $\bU$ and $P_\theta$ can be constructed from a function $\chi(\bu,\theta)$. 

\begin{proposition}
\label{thm1}
Let $\bX$ be a random vector with density $f_\bX(\bx)$ and support $\cX$. 
Let further $\chi(\bu,\theta)$ for $\bu \in \cU$, $\theta \in \Omega$ be such that $\chi(\bu,\theta)$ for each fixed $\theta \in \Omega$ has a range that contains  $\cX$, is differentiable, and is one-to-one with a continuous inverse. Let $\bU$ be a random vector taking values in $ \cU$, with distribution depending on $\theta \in \Omega$ and given by the density 
\beq
\label{fu}
  f(\bu\;|\;\theta) = f_\bX(\chi(\bu,\theta)) \left|\det \partial_\bu \chi(\bu,\theta) \right|.
  \eeq
 Then 
 (\ref{pivot}) holds.
\end{proposition}

\begin{proof} Let $\phi$  be an arbitrary bounded function and fix a $\theta$.  Then by a standard change of variable argument \citep[Theorem 7.26]{rudin}  we have 
\begin{eqnarray*}
  \E[  \phi(\chi(\bU,\theta))] &=& \int \phi(\chi(\bu,\theta)) f(\bu\; |\; \theta) d\bu \\
  &=& \int \phi(\chi(\bu,\theta)) f_\bX(\chi(\bu,\theta))\cdot  \left|\det \partial_\bu \chi(\bu,\theta) \right| d\bu \\
   &=& \int \phi(\bx) f_\bX(\bx) d\bx\\
   &= &\E[\phi(\bX)].
  \end{eqnarray*}
  The result of the proposition then holds since $\phi$ was arbitrarily chosen. 
  \end{proof}
  
 \bigskip

We now introduce the following assumption:

\begin{assumption}
For any $\bu \in \cU$ and $\theta \in \Omega$, the equation $\tau(\bu,\theta)=t$ can be uniquely solved for $\theta$ by  $\theta= \hat \theta(\bu,t)$.
\end{assumption}  

\bigskip

In order to derive the conditional distribution of $\bX$ given $T(\bX)=\bt$, we will consider conditional expectations of a function $\phi$. Under Assumption~1 we have
\begin{eqnarray}
\E[\phi(\bX)|T(\bX)=\bt] &=& \nonumber \E[\phi(\chi(\bU,\Theta))|\tau(\bU,\Theta)=\bt] \\
&=& \E[\phi(\chi(\bU,\hat \theta(\bU,\bt)))|\tau(\bU,\Theta)=\bt], \label{comp}
\end{eqnarray}
where we used the substitution principle \citep{bahadurbickel} noting that 
$\tau(\bU,\Theta)=\bt \Leftrightarrow \Theta = \hat \theta(\bU,\bt)$. In order to calculate (\ref{comp}) we will hence need the conditional distribution of $\bU$ given $\tau(\bU,\Theta)=\bt$. This distribution is obtained from a standard transformation from $(\bU,\Theta)$ to $(\bU,\tau(\bU,\Theta))$, which gives the joint density $h(\bu,\bt)$ of $(\bU,\tau(\bU,\Theta)$ as
\beq
\label{huttetu}
 h(\bu,\bt) = f(\bu |  \hat \theta(\bu,\bt)) w(\bt,\bu),
 \eeq
 where $\bt \mapsto w(\bt,\bu)$ is the density of $\tau(\bu,\Theta)$ for fixed $\bu$. This density is given by
 \beq
 \label{wtu}
  w(\bt,\bu) =  \pi(\hat \theta(\bu,\bt)) \left|\det \partial_t \hat \theta(\bu,\bt)\right|
  =
  \left| \frac{\pi(\theta)}{\det \partial_\theta \tau(\bu,\theta)} \right|_{\theta=\hat \theta(\bu,\bt)},
 \eeq
where $\pi(\theta)$ is the density of $\Theta$. 
 From this we get the conditional distribution of $\bU$ given $\tau(\bU,\Theta)=\bt$  as $h(\bu|\bt) = h(\bu,\bt)/\int h(\bu,\bt) d\bu$, and we are then in a position to complete the calculation of (\ref{comp}):
\begin{eqnarray}
   \E[\phi(\bX)|\bT=\bt] &=& \E[\phi(\chi(\bU,\hat \theta(\bU,\bt)))|\tau(\bU,\Theta)=t] \nonumber \\
      &=& \int \phi(\chi(\bu,\hat \theta(\bu,\bt)) h(\bu|\bt) d\bu  \nonumber\\
      &=& \frac{\int \phi(\chi(\bu,\hat \theta(\bu,\bt))) h(\bu,\bt) d\bu }
      {\int  h(\bu,\bt) d\bu } . \label{efi}
      \end{eqnarray}

\subsection{Methods of computation and simulation from the conditional distribution}
The integrals in (\ref{efi}) will usually have an intractable form. The calculation of (\ref{efi}) or simulation of samples from $h(\bu|\bt)$, may hence be done by suitable numerical techniques. Some approaches are briefly considered below. 

\subsubsection{Importance sampling} Importance sampling appears to be the traditional method used in conditional Monte Carlo, see for example \cite{dubi}.  Consider the computation of (\ref{efi}).  If $\bU$ is sampled from a density $g(\bu)$, then (\ref{efi}) can be written
\[
     \E[\phi(\bX)|\bT=\bt]  = \frac{\E[\phi(\chi(\bU,\hat \theta(\bU,\bt)))h(\bU,\bt)/g(\bU)]}{\E[h(\bU,\bt)/g(\bU)]},
\]
which in principle is straightforward to calculate by Monte Carlo simulation.

\subsubsection{Rejection sampling}
\label{212} In order to obtain samples from the conditional distribution of $\bX$ given $\bT=\bt$, we need to first sample $\bU=\bu$ from a density proportional to $h(\bu,\bt)$, then solve the equation $\tau(\bu, \theta) = \bt$, and finally return the conditional sample $\hat \bx = \chi(\bu,\hat \theta(\bu,t))$. Let $\tilde h(\bu,\bt)$ be proportional to $h(\bu,\bt)$ as a function of $\bu$.  In rejection sampling \citep[p. 60]{RIPLEY} one samples from a density $g(\bu)$ with support which includes the support of $\tilde h(\bu,\bt)$ and for which we can find a bound $M<\infty $ such that $\tilde h/g \le M$.  One then samples $\bu$ from $g$ and a uniform random number $z \in[0,1]$ until $Mz \le \tilde h(\bu)/g(\bu)$. 

\subsubsection{Markov Chain Monte Carlo}
\label{213}
 A disadvantage of rejection sampling is the need for the bound $M$ which may be difficult to obtain. The  Metropolis-Hastings algorithm \citep{hastings} needs no such bound but, on the other hand, produces dependent samples. We describe below an approach where proposals of the  Metropolis-Hastings algorithm are independent samples $\bu$ from a density $g(\bu)$, where $g$, as for the rejection sampling method, needs to have a support which includes the support of $\tilde h(\bu,\bt)$. 

To initialize  the algorithm one needs an initial sample $\bu^{0}$ with $\tilde h(\bu^{0}, \bt) > 0$. Then for each iteration $k$, one generates (i) a proposal $\bu^{'}$
from $g(\cdot)$ ; (ii)  a uniform random number $z \in [0,1]$. One then accepts the proposal and let $\bu^{k+1}= \bu^{'}$ if 
\beq
\label{mcmc}
    z \le \frac{\tilde h(\bu^{'},\bt)}{\tilde h(\bu^{k},\bt)}
    \cdot  \frac{g(\bu^{k})}{g(\bu^{'})},
\eeq
and otherwise lets $\bu^{k+1}= \bu^{k}$. It should be noted that for each new proposal $\bu^{'}$ one needs to solve the equations leading to $\hat \theta(\bu^{'},\bt)$. As for  rejection sampling, one obtains the desired samples $\hat \bx^{k} = \chi(\bu^k,\hat \theta(\bu^k,\bt))$.

\subsubsection{The naive sampler}
\label{naive}
In order to check algorithms for conditional sampling, a type of benchmark might be to use a naive sampler as follows. Then $\bx$  are sampled from $f_\bX(\bx)$ and accepted if and only if $|T(\bx) - \bt| < \epsilon$ for an apriori chosen (small) $\epsilon >0$ and an appropriate norm $|\cdot|$. The successive accepted samples $\hat \bx$ are approximate samples from the desired conditional distribution, see Section~\ref{twopar} for examples.

\section{Application of the method}
\label{sec3}

As might be clair from the previous section, the choice of the function $\chi(\bu,\theta)$ and the marginal distribution for $\Theta$ are of crucial importance for the construction of an efficient algorithm.

\subsection{The choice of the function $\chi(\bu,\theta)$} 
The choice of $\chi(\bu,\theta)$ will obviously depend very much on the application, and we refer to the examples in order to give some advice here. The uniqueness requirement of Assumption~1 of course restricts considerably the choice.  An important issue is the requirement that the range of $\chi(\bu,\theta)$, for each $\theta$,  should include the support of $\bX$.  A further discussion on the form of the function  $\chi(\bu,\theta)$ is found in the concluding remarks of Section~\ref{sec7}. In particular is considered a possible relaxation of the uniqueness requirement of Assumption~1. 

\subsection{The choice of distribution of $\Theta$}
\label{thetat}

In Bayesian statistics it is well recognized that prior distributions for parameters may be chosen as improper distributions. Also, in the approach on conditional sampling given sufficient statistics in \cite{LT05} it was argued that a  random vector similar to our $\Theta$ may sometimes preferably be given an improper distribution. The following argument shows, however, that in the present approach, $\Theta$ must  be given a proper distribution (i.e., having an integrable density function $\pi(\theta)$). 

Suppose namely that $\Theta$ is given an improper distribution. In order  to condition on $\tau(\bU,\Theta)$ it is necessary that its density is finite. (In Bayesian analysis, this is the marginal density of the data which appears in the denominator of Bayes' formula.) This property implies that there is a set $A$ such that $P(\tau(\bU,\Theta) \in A) < \infty$, where this set clearly may be chosen so that $P(T(\bX) \in A)>0$. Now for this set $A$,
\begin{eqnarray*}
P(\tau(\bU,\Theta) \in A) &=& \int_\Omega P(\tau(\bU,\Theta) \in A\;|\;\Theta = \theta) \pi(\theta) d\theta \\
&=& \int_\Omega P(\tau(\bU,\theta) \in A\;|\;\Theta = \theta) \pi(\theta) d\theta \\
&=& P(T(\bX) \in A)\int_\Omega \pi(\theta) d\theta,
\end{eqnarray*} 
where the last equality follows from the basic property (\ref{pivot}). This clearly implies that $\int_\Omega \pi(\theta) d\theta < \infty$. 

In the following we shall therefore always assume that $\Theta$ is given an integrable density $\pi(\theta)$. (This density may of course be normalized to have integral one, but as we shall see in our examples, the normalizing constant is usually of no concern). Particular choices will depend on the application, but also on certain structural issues of the problem as explained in the next subsection. 

\subsection{The ``pivotal'' condition}
\label{piv}
In some cases it is possible to choose the function $\chi(\bu,\theta)$ in such a way that $\tau(\bu,\theta)$ depends on $\bu$ only through a lower dimensionable function $r(\bu)$, where the value of $r(\bu)$ can be uniquely
recovered from the equation $\tau(\bu,\theta) = \bt$ for given $\theta$ and $\bt$. This means that there is a function
$\tilde{\tau}$ such that $\tau(\bu,\theta) = \tilde{\tau}(r(\bu),\theta)$ for all $(\bu,\theta)$, and a function $\tilde{v}$
such that $\tilde{\tau}(r(\bu),\theta)=\bt$ implies $r(\bu)=\tilde{v}(\theta,\bt)$. 
The notion of ``pivotal'' for the present case is borrowed from \cite{LT05}, who considered a similar condition in which case  
 $\tilde{v}(\theta,\bT)$ is a pivotal quantity in the classical meaning of the notion. Although the setting here is different, we shall keep calling this the \textit{pivotal} condition.

Under  Assumption 1, the following equivalences hold under the pivotal condition:
\[
   \hat \theta(\bu,\bt)=\theta \; \Leftrightarrow  \;  \tau(\bu,\theta)=\bt  \; 
   \Leftrightarrow  \;  \tilde\tau(r(\bu),\theta)=\bt  \; 
   \Leftrightarrow   \; r(\bu) = \tilde v(\theta,\bt)
   \]
We hence have the identity
 \beq
   \nonumber
       r(\bu) = \tilde v(\hat \theta (\bu,\bt),\bt) \; \mbox{   for all $\bu,\bt$}
   \eeq 
so that 
$$\tau(\bu,\theta) = \tilde{\tau}(\tilde{v}(\hatt(\bu,\bt),\bt),\theta)$$
 and hence
$$\mbox{det} \;
\partial_\theta \tau(\bu,\theta) =
\mbox{det} \; \partial_\theta \tilde{\tau}(\tilde{v}(\hatt(\bu,\bt),\bt),\theta).$$ 
Substituting $\hatt(\bu,\bt)$ for $\theta$ it is
therefore seen that
\begin{equation}
\label{sufcond2} |\det \partial_\theta \tau(\bu,\theta)|_{\theta \lik \hatt(\bu,\bt)}
  \lik  J(\hat{\theta} (\bu, \bt),\bt)
\end{equation}
where
\[
        J(\theta,\bt) \lik  |\mbox{det} \partial_\theta \tilde{\tau}
        (\bv,\theta) |_{\bv \lik \tilde{v}(\theta,\bt)}.
\]
Consider first the case where $J(\hat{\theta} (\bu, \bt),\bt)$ factors as $K(\hat{\theta} (\bu, \bt))a(\bt)$. Suppose also that $f(\bu|\theta)$ factors as
\[
   f(\bu|\theta) = \rho(\theta) \tilde f(\bu|\theta).
\] 
Then (\ref{huttetu}) and (\ref{wtu})  suggest the choice of $\pi(\theta)$ proportional to $\rho(\theta)^{-1}K(\theta)$, which simplifies the expression for $h(\bu,\bt)$ in (\ref{huttetu}). In order to ensure that $\pi(\theta)$ has a finite integral, we might in addition restrict the support of $\pi$ to some bounded set, letting for example
\[
   \pi(\theta) = \rho(\theta)^{-1}K(\theta) I(\theta \in A),
\]
where $I(\cdot)$ is the indicator function of the condition in the parantheses, and $A$ is such that $\int_A \rho(\theta)^{-1}K(\theta)d\theta < \infty$. It follows in this case that
\beq
\label{hugitt}
    h(\bu,\bt) \propto \tilde f(\bu|\hat \theta(\bu,\bt)) I(\hat \theta(\bu,\bt) \in A).
\eeq
For the general pivotal case, leading to (\ref{sufcond2}), we would have to choose a $\pi(\cdot)$ that depends on $\bt$, by replacing $K(\theta)$ by 
$J(\theta,\bt)$ in the above. 
Since $\bt$ is fixed when conditioning on $\bT=\bt$, it is seen that the crucial arguments will go through also in this case, thus still leading to (\ref{hugitt}).  A similar argument was used in \cite{LT05}. 

As a further refinement, it may happen that $r(\bU)$ is a sufficient statistic in the model defined by $f(\bu|\theta)$. Then by Neyman's factorization criterion \citep[Ch. 6]{casella}, we can write
\[
   f(\bu|\theta) = p(r(\bu)|\theta) q(\bu) 
\]
for appropriate functions $p$ and $q$. Hence we can write
\[
   f(\bu|\hat \theta(\bu,\bt)) = p(\tilde v(\hat \theta(\bu,\bt),\bt)|\hat \theta(\bu,\bt)) q(\bu) 
\]
By assimilating the $p(\cdot)$-part of the above into $\pi(\hatt(\bu,\bt))$ (where $\pi(\cdot)$ will now possibly depend on $\bt$) we get
\beq
\nonumber
    h(\bu,\bt) \propto q(\bu) I(\hat \theta(\bu,\bt) \in A).
\eeq

\section{Examples}
\label{sec4}

\subsection{Two examples involving the pivotal condition}

\subsubsection{Conditional sampling of uniforms} 
\label{411}
Let $\bX=(X_1,X_2,\ldots,X_n)$ be an i.i.d sample from $U[0,1]$, the uniform distribution on $[0,1]$,   and let $T(\bX)=\sum_{i=1}^n X_i$. Suppose one wants to sample from the conditional distribution of $\bX$ given $T(\bX)=t$ where $0<t<n$.  There appears to be no simple expression for this conditional distribution. \cite{LT05} considered an approach where the uniform distribution is embedded in a parametric family involving truncated exponential distributions and utilzed the sufficiency of $T(\bX)$ in this model.  The resulting method is, however, surprisingly complicated in absence of the pivotal condition of \cite{LT05}. A Gibbs sampling method was devised by \cite{rannestad}, apparently being much quicker than the former method, and much easier to implement. 

We now present a simple solution to the problem using the approach of the previous sections and utilizing the presence of a pivotal condition as studied in Section~\ref{piv}. An advantage as compared to the Gibbs sampling algorithm is that the present method produces independent samples.

Let $\bU=(U_1,U_2,\ldots,U_n)$ be an i.i.d.\ sample from  $U[0,\theta]$,  where $\theta \in (0,1]$. Then the $U_i/\theta$ are i.i.d.\ from $U[0,1]$, so  condition (1) in Section~\ref{sec2} is satisfied with 
\beq
\label{chiq}
\chi(\bu,\theta) = \left( \frac{u_1}{\theta},\frac{u_2}{\theta},\ldots, \frac{u_n}{\theta} \right).
\eeq 
defined for $\bu \in [0,1]^n$ and $\theta \in (0,1]$. 

The above is moreover in accordance with Proposition~\ref{thm1}, which readily gives 
\[
    f(\bu\;|\;\theta) = \frac{1}{\theta^n} I(\max_i u_i \le \theta),
\] 
where we used that $f_\bX(\bx) = I(\max x_i \le 1)$. Note that  here and below we tacitly assume that we are working with nonnegative variables only.
 
Now we have
\[
   \tau(\bu,\theta) = 
    \frac{\sum_{i=1}^n u_i}{\theta},
\]
and hence there exists a unique solution for $\theta$ of the equation $\tau(\bu,\theta)=t$, given by
\beq
\label{thatq}
    \hat \theta(\bu,t) = \frac{\sum_{i=1}^n u_i}{t}.
\eeq
Clearly, the pivotal condition of Section~\ref{piv} is satisfied with $r(\bu)=\sum_{i=1}^n u_i$. It is seen, however, that $r(\bu)$ does not correspond to a sufficient statistic for the model $f(\bu|\theta)$. We should therefore stick to (\ref{hugitt}), which by choosing $A=[0,1]$ gives 
\begin{eqnarray}
     h(\bu,t) &\propto &  I(\max_i u_i \le \hatt(\bu,\bt)) \cdot I(\hatt(\bu,\bt) \le 1)
     \nonumber \\
     & =& I\left(\max u_i \le \frac{\sum_{i=1}^n u_i}{t}\right)
    \cdot I\left(\frac{\sum_{i=1}^n u_i}{t}\le 1\right)  \nonumber \\
		&=& I\left(t \cdot \max_i u_i \le
    \sum_{i=1}^n u_i \le t\right). \label{condition}
    \end{eqnarray}
Thus  the conditional distribution $h(\bu| t)$ is uniform on the set of $\bu \in [0,1]^n$ satisfying the restriction given by the indicator function in (\ref{condition}).  We may hence sample the $u_i$ independently from $U[0,1]$ and accept the sample if and only if the restriction
		is satisfied. (Note that if $t \le 1$, then the left inequality in (\ref{condition}) is always satisfied.)
Finally, for the accepted samples we conclude from (\ref{chiq}) and (\ref{thatq}) that the resulting conditional sample is
\beq
\label{finals}
\hat \bx =    \left( t \; \frac{u_1}{\sum_{i=1}^n u_i}, t \; \frac{u_2}{\sum_{i=1}^n u_i} , \ldots,t \; \frac{u_n}{\sum_{i=1}^n u_i} \right).
\eeq

Figure~\ref{figure2} shows the result of a simulation with $n=2$ and $t=0.3$. It is easy to show by direct calculation that the conditional distribution of $X_1$ (and hence also of $X_2$) given $X_1+X_2=0.3$ is uniform on $[0,0.3]$. The left panel of the figure shows the empirical cumulative distribution of $X_1$ (and $X_2$) resulting from (\ref{finals}), which is clearly uniform on $[0,0.3]$ as expected. The right panel of the figure shows, on the other hand, the empirical distribution obtained from (\ref{finals}) when  sampling $(u_1,u_2)$ i.i.d.\ from $U[0,1]$ without ignoring the pairs $(u_1,u_2)$ with $u_1 + u_2 > 0.3$ (which is required by (\ref{condition})). 
\begin{figure}
		\centering
		\includegraphics[width = 8cm]{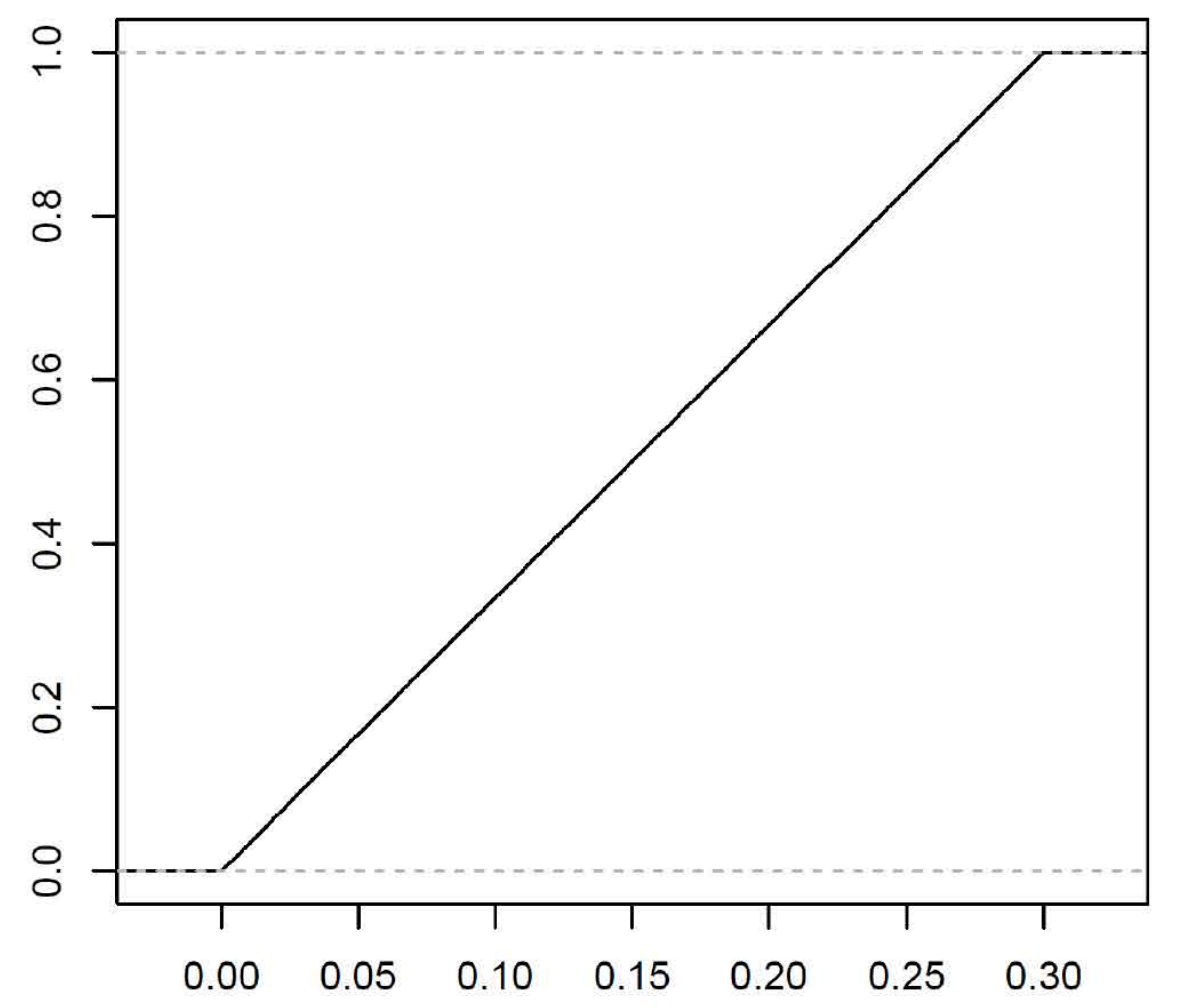}
\includegraphics[width=8cm]{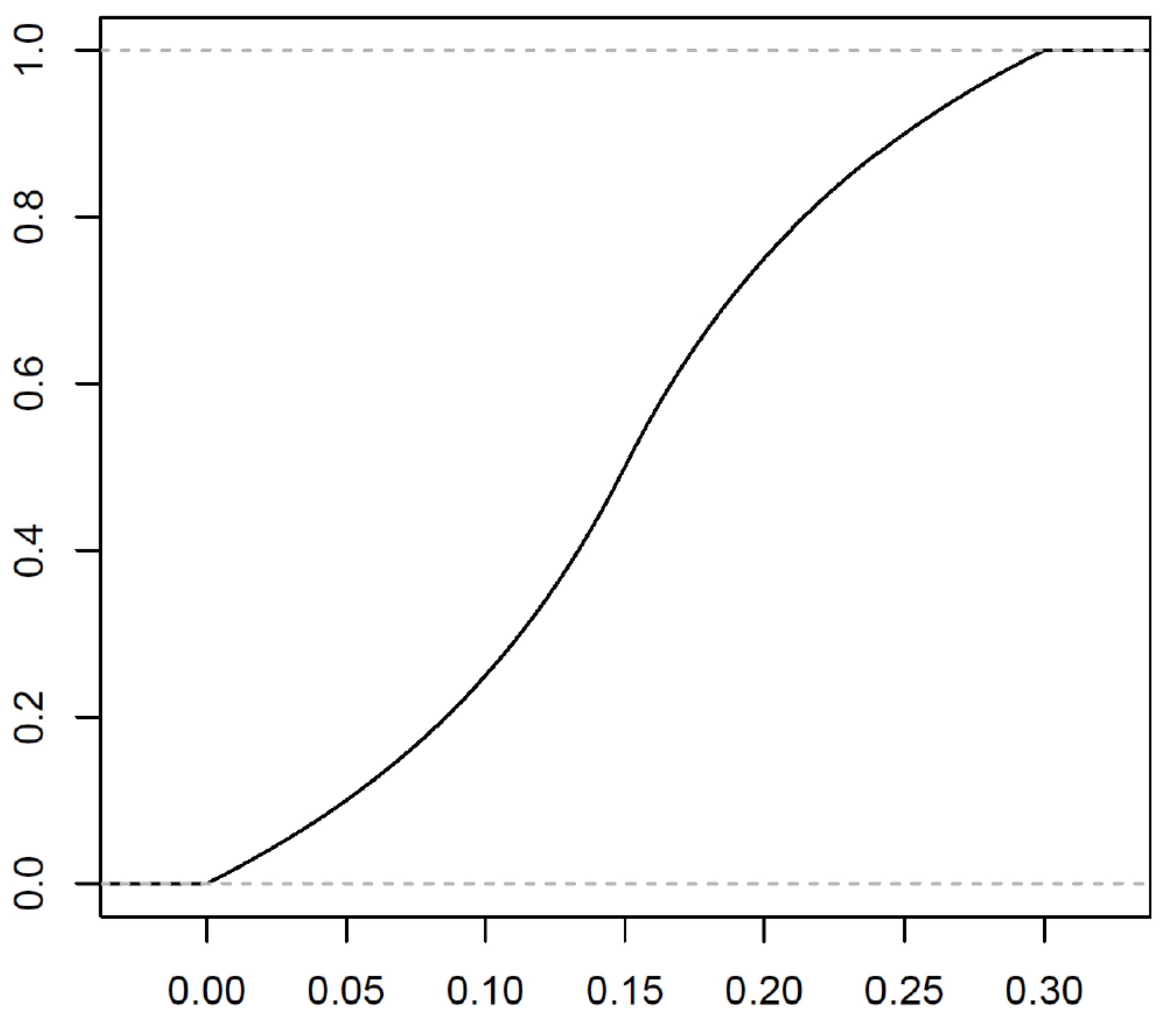} 
\caption{Empirical distribution functions for marginal distributions of conditional samples of uniforms when $n=2$. Left: Sampling $(u_1,u_2)\sim U[0,1]$ and using (\ref{condition}) and (\ref{finals}). Right: Sampling $(u_1,u_2)\sim U[0,1]$ and using (\ref{finals}) only. }	
	\label{figure2}
\end{figure}
The discrepancy from a straight line shows  that  condition (\ref{condition}) is necessary here.
Still the algorithm is very simple, and simpler than the corresponding algorithms of \cite{LT05} and \cite{rannestad} that were mentioned above.

The algorithm may be slow if $t$ is close to 0 or $n$. In these cases it might be better to use importance sampling by drawing the $u_i$ from a density $g(u) = c u^{c-1}$ for $c>0$, where $c$ is small (large) if $t$ is close to $0$ (close to $n$). But note that we will then need to sample from a non-uniform density $h(\bu|t)$.

As a final remark on this example, suppose instead that we wanted to condition on $\sum_{i=1}^n X_i^r=t$ for some given $r>0$. It is then straightforward to check that only a minor modification of the above derivation is needed. As a result, one should still sample $u_i$ from $U[0,1]$, but change condition (\ref{condition}) into
\[
I\left(t \cdot \max_i u_i^r \le
    \sum_{i=1}^n u_i^r \le t\right)
\] 
and use the samples $\hat \bx$ where
\[
\hat x_i =  t^{1/r} \; \frac{u_i}{(\sum_{\ell=1}^n u_\ell^r)^{1/r}}.
\]

\subsubsection{Conditional sampling of normals}
\label{412}
The following is a classical example in conditional Monte Carlo, see e.g.\ \cite{tukey}, \cite{hammersley}, \cite{grano}, \cite{RIPLEY}. 
Let $\bX=(X_1,X_2,\ldots , X_n)$ be i.i.d from $N(0,1)$ and let $T(\bX)=\max_i X_i - \min_i X_i$. We wish to sample from the conditional distribution of $\bX$, given $T(\bX)=t$ for $t>0$. 

Now let $\bU=(U_1,U_2,\ldots,U_n)$ be an i.i.d.\ sample from  $N(0,\theta^2)$. Then condition (1) in Section~\ref{sec2} is clearly satisfied when $\chi(\bu,\theta)$ is given by the scale transformation (\ref{chiq}) for $\bu=(u_1, u_2, \ldots , u_n) \in \RealN^n$ and $\theta \in (0,1]$. It is furthermore  seen that the pivotal condition of Section~\ref{piv} is satisfied with $r(\bu) = \max_i u_i - \min_i u_i$, and in a similar way as for the uniform distribution case treated above, we arrive at 
\begin{equation}
\label{norm}h(\bu ,t)\propto \exp\left(-{\frac{t^2}{2(\max_i u_i-\min_i u_i)^2}} \sum_{i=1}^n u_i^2\right)I(\max_i u_i - \min_i u_i < t),
\end{equation}
	by letting $\pi(\cdot)$ be supported on the interval $[0,1]$. Actually, we have used $\pi(\theta)=\theta^{n-1}I(0<\theta \le 1)$.

Noting that the right hand side of (\ref{norm}) is less than or equal to 
\[
\exp\left(-\frac{1}{2} \sum_{i=1}^n u_i^2\right)I(\max_i u_i - \min_i u_i < t)
\]
we can use rejection sampling (Section~\ref{212}) based on sampling of i.i.d. standard normal variates. If $t$ is small, then in order to increase the acceptance probability of the rejection sampling, it might be beneficial to use as the proposal distribution, a mixture of a standard normal and a normal distribution with small variance. 

The resulting conditional samples  are now of the form
\beq
\nonumber
\hat \bx = 
    \left( t \; \frac{u_1}{\max_i u_i-\min_i u_i},   \ldots, t\; \frac{u_n}{\max_i u_i-\min_i u_i} \right).
\eeq

\subsection{Conditional sampling from  two-parameter exponential families}
\label{twopar}

Suppose $\bX=(X_1,X_2,\ldots,X_n)$ is distributed as an i.i.d.\ sample from a two-parameter exponential family of \textit{positive} random variables, with minimal sufficient statistic 
\beq
\label{kanskje}
\bT(\bX) = (T_1(\bX),T_2(\bX)) = \left(\sum_{i=1}^n g_1(X_i), \sum_{i=1}^n g_2(X_i)\right).
\eeq
Suppose now that $\bt=(t_1,t_2)$ is the observed value of $\bT(\bX)$, and that we want to sample $\bX=(X_1,X_2,\ldots,X_n)$ conditionally on $\bT(\bX)=\bt$. By sufficiency, samples from the conditional distribution of $\bX$ given $\bT(\bX) =\bt $ can be obtained by choosing any density from the given family as the basic density. Let $f_\bX(\bx)=\prod_{i=1}^n f_X(x_i)$ be the chosen density. 

Let 
\beq
\label{kjiexp}
    \chi(\bu,\theta) = \left(  \left( \frac{u_1}{\beta} \right)^\alpha, \left(\frac{u_2}{\beta} \right)^\alpha,  \ldots,  \left( \frac{u_n}{\beta} \right)^\alpha \right),
\eeq
where $\bu=(u_1,u_2,\ldots,u_n)$ is a vector of positive numbers and $\theta=(\alpha,\beta)$ is a pair of positve parameters $\alpha, \beta$. 
Then, using Proposition~\ref{thm1}, condition (\ref{pivot}) of Section~\ref{sec2} is satisfied if $\bU$ for given $\theta$ has density
\beq
\label{futh}
f(\bu\;|\;\theta) = \prod_{i=1}^n \frac{\alpha}{\beta}  \left( \frac{u_i}{\beta} \right)^{\alpha-1} f_X\left( \left(\frac{u_i}{\beta}\right)^\alpha\right).
\eeq
Assumption~1 requires that there is a unique solution for $\theta$ of the equation
\[
    \tau(\bu,\theta) = \bt,
    \]
 which here means
    \[
     \sum_{i=1}^n g_1 \left( \left( \frac{u_i}{\beta} \right)^\alpha \right) = t_1,
     \]
     \[
     \sum_{i=1}^n g_2 \left( \left( \frac{u_i}{\beta} \right)^\alpha \right) = t_2.
     \]
Assume that there is a unique solution $\hat \theta(\bu,\bt) = (\hat \alpha(\bu,\bt), \hat \beta (\bu,\bt))$ of these equations.

Letting $\pi(\theta)\equiv \pi(\alpha,\beta)$ be the density of $\Theta$, 
the density $h(\bu,\bt)$ is found from (\ref{huttetu}) and (\ref{wtu}), giving 
 \begin{eqnarray} \nonumber
  h(\bu,\bt)& =& f(\bu\;|\;\hat \theta(\bu,\bt)) w(\bt,\bu) \\
  &=&  \frac{(\hat \alpha/\hat \beta)^n \left(\prod_{i=1}^n \hat x_i \right)^{1-1/\hat \alpha} \left( \prod_{i=1}^n f_X(\hat x_i)\right) \pi(\hat \alpha,\hat \beta)}
           {|\det \partial_\theta \tau(\bu,\theta) |_{\theta=\hat \theta(\bu,\bt)} }. \label{hutex}
\end{eqnarray}
where 
\beq
   \label{xhat} 
     \hat x_i = \left( \frac{u_i}{\hat \beta} \right)^{\hat \alpha}
    \eeq
and
   \begin{eqnarray*}
   \det \partial_\theta \tau(\bu,\theta) |_{\theta=\hat \theta(\bu,\bt)} 
   &=&  \frac{1}{\hat \beta(\bu,\bt)} \left[ \left(\sum_{i=1}^n g_1'(\hat x_i)\hat x_i \right) \left(\sum_{i=1}^n g_2'(\hat x_i)\hat x_i \log(\hat x_i)\right) \right. \\
  & -&  \left. \left(\sum_{i=1}^n g_2'(\hat x_i)\hat x_i \right) \left(\sum_{i=1}^n g_1'(\hat x_i)\hat x_i \log(\hat x_i)\right) \right].
   \end{eqnarray*}

When sampling from (\ref{hutex}) by the Metropolis-Hastings algorithm (Section~\ref{213}) it seems to be a good idea to let the proposal distribution $g(\bu)$ be the distribution of the original exponential familiy with parameter values equal to the maximum likelihood estimates based on the observation $\bt$. Then the calculated $\hat \alpha,\hat\beta$ are expected to be  around 1, and we therefore suggest to choose the distribution of $\Theta$  as $\pi(\alpha,\beta)= I( a_1 \le \alpha \le a_2, b_1 \le \alpha \le b_2)$ for suitably chosen $0<a_1 < 1 < a_2$, $0<b_1 < 1 < b_2$, see examples in Section~\ref{sec5}. 

In a practical application one would usually also have the original data $\bx = (x_1,\ldots,x_n)$ which led to the values $t_1=T_1(\bx), t_2=T_2(\bx)$. The vector $\bx$ may then be used as the initial sample of the Metropolis-Hastings simulation, and will give $\hat \alpha = \hat \beta = 1$.  In this case, the successively simulated accepted conditional samples 
$\hat \bx = (\hat x_1,\ldots,\hat x_n)$ defined by (\ref{xhat}) will have the correct distribution, so there is no need for a burn-in period in the Metropolis-Hastings simulations. 
    
\subsubsection{Gamma Distribution}
\label{421}
The gamma-distribution with shape parameter $k>0$ and scale parameter $\theta>0$ has 
 density
\beq
\label{gammaden}
   f(x;k,\theta) =  \frac{1}{\theta^k \Gamma(k)} x^{k-1}e^{-x/\theta} \mbox{ for $x>0$}.
    \eeq
    We suggest using $k=\theta=1$ to get $f_X(x)=e^{-x}$. Referring to (\ref{kanskje}), we have for the gamma model, $g_1(x)=x, \; g_2(x)=\log x$, and hence we need to solve the equations
\begin{eqnarray}
\sum_{i=1}^n   \left( \frac{u_i}{\beta} \right)^\alpha & =& t_1, \label{gamma1}\\
\sum_{i=1}^n \log  \left( \frac{u_i}{\beta} \right)^\alpha  &=& t_2. \label{gamma2}
\end{eqnarray}
It is shown in Lemma~\ref{unique1} in Appendix that there is a unique solution $(\hat \alpha,\hat \theta)$ for $(\alpha,\theta)$. The actual solution turns out to be easily  obtained via a single equation involving~$\alpha$. Now  
(\ref{hutex}) gives
\begin{equation}\nonumber
   h(\bu,\bt) = \frac{(\hat \alpha/\hat \beta)^n e^{(1-1/\hat \alpha)t_2} e^{-t_1} \pi(\hat \alpha,\hat \beta)}
           {(1/\hat \beta)\left(t_1t_2-n \sum_{i=1}^n \hat x_i \log \hat x_i\right)},
           \end{equation}
           which is the basis for simulation of conditional samples as outlined above. It is easy to see, however, that the pivotal condition of Section~\ref{piv} is not satisfied here.

      \subsubsection{Inverse Gaussian Distribution}      
The Inverse Gaussian distribution has density which can be written as 
 \beq
 \label{igdensity}
 f (x; \mu, \lambda) = \sqrt{\frac{\lambda}{2 \pi x^3}} \exp \left( -\frac{\lambda}{2x} -\frac{\lambda 
x}{2\mu^2}+\frac{\lambda}{\mu} \right), \;\; x 
> 0.
\eeq
Let now $f_X(x)$ be the density obtained when $\mu=\lambda=1$, i.e.
\[
 f_X (x) = \sqrt{\frac{1}{2 \pi x^3}} \exp \left( -\frac{1}{2x} -\frac{x}{2}+1 \right), \;\; x 
> 0.
\]
Furthermore, for the inverse Gaussian distributions we can choose $g_1(x)=x, \; g_2(x)=1/x$ \citep[p. 7]{seshadri}, and hence we obtain the equations
\begin{eqnarray*}
\sum_{i=1}^n   \left( \frac{u_i}{\beta} \right)^\alpha  &=& t_1, \\
\sum_{i=1}^n  \left( \frac{u_i}{\beta} \right)^{-\alpha} & =& t_2.
\end{eqnarray*}
As for the gamma case, there is a unique solution $(\hat \alpha,\hat \beta)$ for $(\alpha,\beta)$, see Lemma \ref{unique2} in the Appendix. Now we get from (\ref{hutex}),
\begin{equation}\nonumber
   h(\bu,\bt) = \frac{(\hat \alpha/\hat \beta)^n \left( \prod_{i=1}^n \hat x_i \right)^{-1/2-1/\hat \alpha} e^{-(1/2)(t_1+t_2)+n} \pi(\hat \alpha,\hat \beta)}
           {(1/\hat \beta) \left(t_2 \sum_{i=1}^n \hat x_i \log \hat x_i - t_1\sum_{i=1}^n  \log \hat x_i/\hat x_i \right)}.
           \end{equation}
It was suggested above to use the parametric model itself as a proposal distribution in Metropolis-Hastings simulations, with parameters given by the maximum likelihood estimates from the original data. Following \cite[p. 7]{seshadri}, the maximum likelihood estimates of the parameters in (\ref{igdensity}) are given from
\[
  \hat \mu = \bar x, \; \; \; \hat \lambda^{-1} = \frac{1}{n} \sum_{i=1}^n \left( \frac{1}{x_i}-\frac{1}{\bar x} \right).
\]
Note also that \cite{michael} presented a nice method of simulating from the inverse Gaussian distribution.

\vekk{ 
\subsubsection{Conditional sampling in von Mises distribution\\ THIS SECTION IS SO FAR NOT CHECKED FOR APPLICABILITY}
The von Mises distribution can be described as the distribution of a point $P$ on the unit circle in $R^2$, represented by its angular coordinate $x$ with respect to the origo $(0,0)$. The von Mises density is
\beq
\nonumber
  f(x; \theta, \kappa) = \frac{1}{2 \pi I_0(\kappa)} \exp\{ \kappa \cos(x-\theta) \} ,
\eeq
where $\theta$ is the location in $[0,2\pi)$,
 $\kappa$ is a positive shape parameter and $I_0(\kappa)$ is the modified Bessel function of order 0. Suppose $\bX=(X_1,\ldots,X_n)$ is an i.i.d.\ sample from this distribution. Then a minimal sufficient statistic is $T(\bX)$ given by
\[
  T_1(\bX) = \sum_{i=1}^n \cos(X_i), \; \; 
   T_2(\bX) = \sum_{i=1}^n \sin(X_i),
\]
which is seen by writing $ f(x; \theta, \kappa) =(2 \pi I_0(\kappa))^{-1} \exp\{ \kappa( \cos \theta \cos x + \sin \theta \sin x)\}$.  
We now wish to sample from the conditional distribution of $\bX$ given $T_1(\bX)=t_1, T_2(\bX)=t_2$. By sufficiency, we can freely choose parameter values, so let us choose $\theta=0,\kappa=1$. We thus have
\[
f_X(x)  = \frac{1}{2 \pi I_0(1)} \exp\{ \cos(x)\}.
\]
The choice of the function $\chi(\cdot,\cdot)$ is now the crucial step. As seen below, we have not succeeded in finding a $\chi(\cdot,\cdot)$ which guarantees unique solutions as required in Assumption~1. The method will thus rely on the approach of Section~\ref{multsol}.

Let 
\beq
\label{grunn}
   \chi(u; \alpha,\theta) = \alpha u - \theta,
\eeq
where $\alpha > 0$, $\theta \in [0,2\pi)$ and $u>0$,
which implies
\[
   f(\bu|\alpha,\theta) = \frac{\alpha^n}{(2 \pi I_0(1))^n}
   \exp\left\{\sum_{i=1}^n\cos(\alpha u_i-\theta) \right\}
\]
For a given $\bu$ we now need to solve for $\alpha$ and $\theta$ the system of equations 
\begin{eqnarray}
\label{tau1}
   \tau_1(\bu;\alpha,\theta) =  \sum_{i=1}^n \cos(\alpha u_i - \theta) &=& t_1 \\ \label{tau2}
  \tau_2(\bu;\alpha,\theta) =  \sum_{i=1}^n \sin(\alpha u_i - \theta) &=& t_2
\end{eqnarray}
In the following we shall make repeated use of the two well known formulas
\begin{eqnarray}
 \cos(u-v) &=& \cos u \cos v + \sin u \sin v \label{cosmin}\\
 \sin(u-v) &=& \sin u \cos v - \cos u \sin v \label{sinmin}
 \end{eqnarray}
It follows immediately that the system (\ref{tau1})-(\ref{tau2}) is equivalent to 
\begin{eqnarray}
\label{AB1}
    A(\alpha) \cos \theta  + B(\alpha)\sin \theta    &=& t_1 \label{p1}\\
  \label{AB2}   B(\alpha) \cos \theta  -A(\alpha) \sin \theta   &=& t_2 \label{p2}
\end{eqnarray}
where
\begin{equation}
   A(\alpha) = \sum_{i=1}^n \cos(\alpha u_i), \; \;
   B(\alpha) = \sum_{i=1}^n \sin(\alpha u_i) . \nonumber
\end{equation}
From (\ref{AB1}) and (\ref{AB2}) we easily obtain
\begin{eqnarray}
\label{kid1}
\cos \theta &=& \frac{A( \alpha) t_1 +  B( \alpha) t_2}
{A^2( \alpha) +  B^2( \alpha)} \\ \label{kid2}
\sin \theta &=& \frac{B( \alpha) t_1 - A( \alpha) t_2}
{A^2( \alpha) +  B^2( \alpha)}
\end{eqnarray}
and by squaring and adding the equations (\ref{kid1}) and (\ref{kid2}) we arrive at 
\beq
\label{sqr}
  A^2(\alpha)  +B^2(\alpha) = t_1^2 + t_2^2
\eeq
or
\[
  \left( \sum_{i=1}^n \cos(\alpha u_i)\right)^2 +
   \left( \sum_{i=1}^n \sin(\alpha u_i) \right)^2 = t_1^2 + t_2^2
\]
Thus we have a single equation for $\alpha$. 
Working out the squares, simplifying using $\sin^2 + \cos^2 =1$,  and again using (\ref{cosmin}) and (\ref{sinmin}), we arrive at the equation
\[
    n + 2 \sum_{i<j} \cos(\alpha(u_i-u_j)) = t_1^2 + t_2^2
\]
or
\beq
\label{ddd}
   \sum_{i<j} \cos(\alpha(u_i-u_j)) = \frac{t_1^2 + t_2^2 - n}{2}
\eeq
In order to check the monotonicity properties of the left hand side of (\ref{ddd}) as a function of $\alpha$, we simulated $n=10$ values of $u_i$ from $N(\pi,1.5)$ and plotted the left hand side of (\ref{ddd}). The result is displayed in Figure~\ref{yz} and shows that there is indeed a possibility of multiple roots. Simulations indicated, on the other hand, that restricting $\alpha$ to $(0,1=$ might be advantageous. 

The solutions for $\theta$ can now be found from (\ref{kid1} or (\ref{kid2}). Note that to each solution for $\alpha$ there will be a different solution for $\theta$. 
\begin{figure}
\centering
\includegraphics[width=9.5cm]{multiple-roots}
\caption{The left hand side of (\ref{ddd}) simulated with $n=10$}
\label{yz}
\end{figure}

From (\ref{tau1}) and (\ref{tau2}) we get 
\begin{eqnarray*} \det
\partial_{\alpha,\theta} \tau(\bu;\alpha,\theta)
&=& \left(\sum_{i=1}^n u_i\sin(\alpha u_i-\theta)\right)
\left(\sum_{i=1}^n \cos(\alpha u_i-\theta)\right)  \\ &-&
\left(\sum_{i=1}^n u_i\cos(\alpha u_i-\theta)\right)
\left(\sum_{i=1}^n \sin(\alpha u_i-\theta)\right)
\end{eqnarray*}
which when applying at a solution $(\hat \alpha,\hat \theta)$ of (\ref{tau1}) and (\ref{tau2}) and using (\ref{cosmin}) and (\ref{sinmin}), becomes
\begin{eqnarray*} \det
\left. \partial_{\alpha,\theta} \tau(\bu;\alpha,\theta)\right|_{(\alpha,\theta)=(\hat \alpha,\hat \theta)}
&=& t_1 \cos\hat \theta\sum_{i=1}^n u_i\sin(\hat \alpha u_i) -
t_1 \sin \hat\theta \sum_{i=1}^n \cos(\hat\alpha u_i)  \\ &-&
t_2 \cos \hat\theta\sum_{i=1}^n u_i\cos(\hat\alpha u_i)
- t_2 \sin \hat\theta \sum_{i=1}^n \sin(\hat\alpha u_i) \\
&=& t_1 \cos\hat \theta\sum_{i=1}^n u_i\sin(\hat \alpha u_i) -
t_2 \cos \hat\theta\sum_{i=1}^n u_i\cos(\hat\alpha u_i) \\
&-& t_1 A(\hat \alpha) \sin \hat\theta 
- t_2 B(\hat \alpha)\sin \hat\theta  \\
&=& t_1 \cos\hat \theta\sum_{i=1}^n u_i\sin(\hat \alpha u_i) -
t_2 \cos \hat\theta\sum_{i=1}^n u_i\cos(\hat\alpha u_i) \\
&-& (t_1^2+t_2^2) \sin \hat\theta \cos \hat\theta 
\end{eqnarray*}
where for the last equality we used (\ref{kid1}) and (\ref{sqr}). Noting that these equations also imply 
\begin{eqnarray*}
\label{kidd1}
\cos \hat \theta &=& \frac{A( \hat\alpha) t_1 +  B(\hat \alpha) t_2}
{t_1^2+t_2^2} \\ \label{kidd2}
\sin \hat \theta &=& \frac{B( \hat\alpha) t_1 -  A(\hat \alpha) t_2}
{t_1^2+t_2^2}
\end{eqnarray*}
it is seen that $\det \left. \partial_{\alpha,\theta} \tau(\bu;\alpha,\theta)\right|_{(\alpha,\theta)=(\hat \alpha,\hat \theta)}$ does not involve $\hat \theta$. This in fact also applies to $f(\bu|\hat \alpha,\hat \theta)$, which becomes 
\[
    f(\bu|\hat \alpha,\hat \theta) = \frac{\hat \alpha^n}{(2\pi I_0(1))^n}
    \exp\{t_1\}
\]
Then, in order to be able to obtain $h(\bu,\bt)$ from (\ref{mult}), we note that the expression can be written without including the calculated $\hat \theta$. 

To finish, suppose $\bu$ is a sample from the density proportional to $h(\bu,\bt)$. The $i$th entry of the resulting conditional sample is, by (\ref{grunn}), 
\[
\hat x_i = \chi(u_i,\hat \alpha,\hat \theta) = \hat \alpha u_i - \hat \theta
\]
which most conveniently can be given by the cosine,
\[
\cos(\hat x_i) = \cos( \hat \alpha u_i)\cos(\hat \theta)
+ \sin( \hat \alpha u_i)\sin(\hat \theta)
\]
} 

\section{A simulation study}
\label{sec5}

A simulation study was performed in order to illustrate the algorithms of 
 Section~\ref{twopar} for the gamma and inverse Gaussian distributions, respectively. The setup of the study is summarized in Table~\ref{table:1}. 

For example, in case 1, a sample $\bx$ with $n=3$ was drawn from a gamma distribution, giving the observed sufficient statistic $(t_1,t_2)=(4.86,1.02)$. Conditional samples were then simulated using the Metropolis-Hastings algorithm in the way described in Section~\ref{twopar}. More precisely, the proposal distribution was chosen to be the gamma density (\ref{gammaden}) using the maximum likelihood estimates $\hat k=3.66,\; \hat \theta=0.44$ as parameters. The density $\pi(\alpha,\beta)$ was chosen to be uniform over $(\alpha, \beta) \in [0.5,1.5] \times [0.5,1.5]$. In addition, we applied  the naive sampling method described in Section~\ref{naive}. Values $\epsilon_1,\epsilon_2$ (see Table~\ref{table:1}) were chosen so that the sampler accepts an i.i.d. sample $\bx'=(x_1', x_2', \ldots , x_n')$ from the proposal distribution if and only if 
\[
|T_1(\bx')-t_1|\leq \epsilon_1\ \text{and}\ |T_2(\bx')-t_2|\leq \epsilon_2.
\]
In case 1 were used $\epsilon_1=\epsilon_2=10^{-2}$. 
Both the Metropolis-Hastings algorithm and the naive sampler were ran for enough iterations to produce at least $10^4$ samples. 

The description is similar for cases 2-4. Figure~\ref{figure1} shows, for each of the four cases in Table~\ref{table:1},  the simulated cumulative distribution functions for the sampled $\hat x_1$. The closeness of the curves corresponding to the two methods is remarkable. Considering the naive sampler as a ``benchmark'', although only approximately correct, this closeness can be taken as a confirmation that the algorithms derived in the paper are producing samples from the correct conditional distributions.

\begin{table}[h!]
	\begin{center}
		\begin{tabular}{||c || c c c c c c ||}  	
			\hline
			Case  & $t_1,t_2$  & $n$ & \text{Distribution} & Sample sizes & $\pi$ & $\epsilon_1, \epsilon_2$\\ 
			\hline\hline
			$1$ & $4.86,1.02$ & $3$ & \text{Gamma} & $10^4$ & $I_{[0.5,1.5]^2}$ & $10^{-2}, 10^{-2}$ \\ 
			\hline
			$2$ & $16.49,2.85$ & $10$  & \text{Gamma} & $10^4$ & $I_{[0.5,1.5]^2}$ & $10^{-1}, 10^{-1}$ \\  
			\hline
			$3$ & $3.67,6.01$ & $3$ & \text{Inverse Gaussian} & $10^4$ & $I_{[0.5,1.5]^2}$ & $10^{-1}, 10^{-1}$ \\  
			\hline
			$4$ & $936.36, 0.59$ & $10$ & \text{Inverse Gaussian} & $10^4$ & $I_{[0.5,1.5]^2}$ & $10^{-1}, 10^{-1}$ \\ 
			\hline
		\end{tabular}
	\end{center}
	\caption{Values used for generating examples.}
	\label{table:1}
\end{table}

\begin{figure}
		\centering
		\includegraphics[width = 6.7cm]{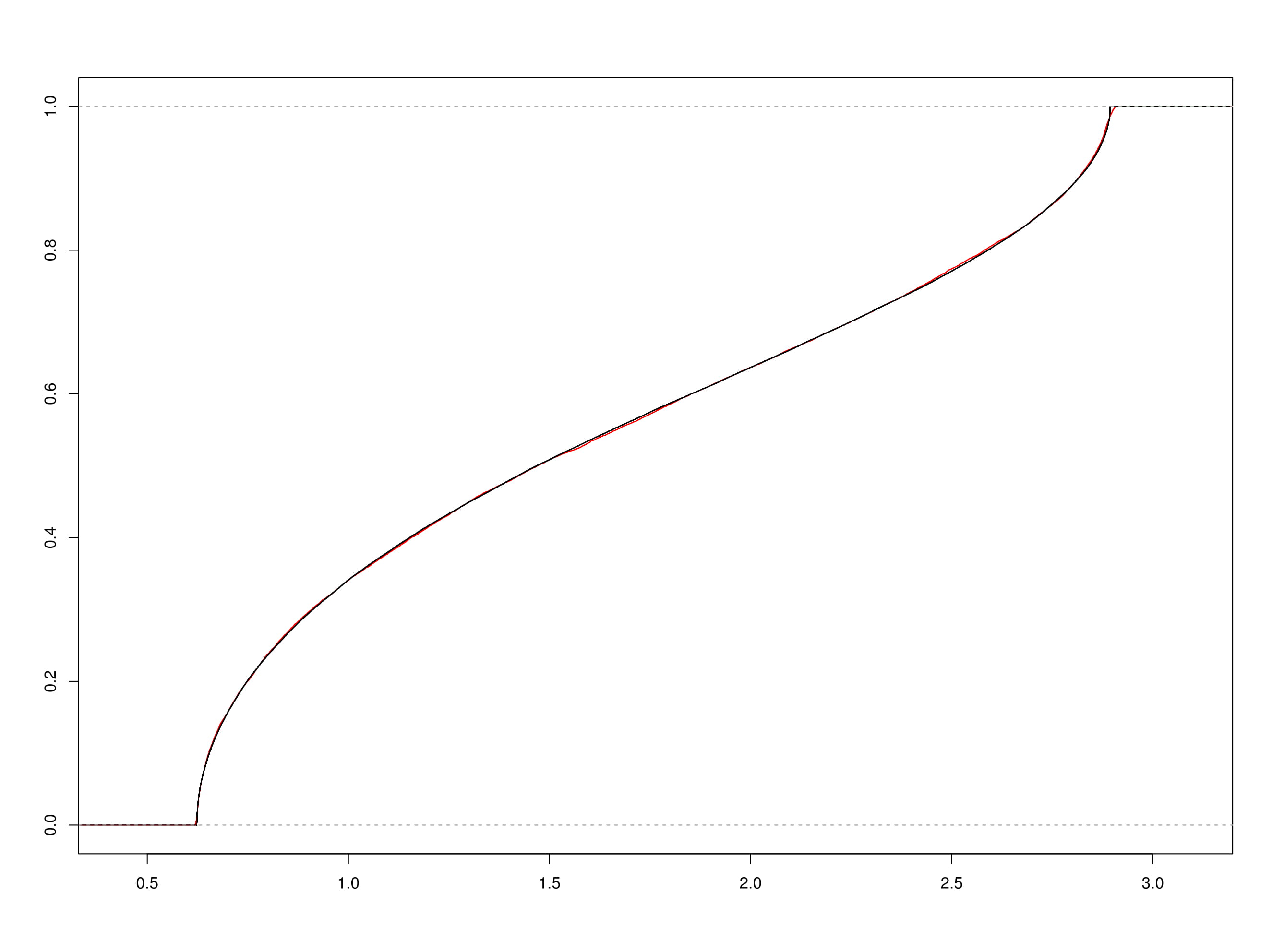}
\includegraphics[width=6.7cm]{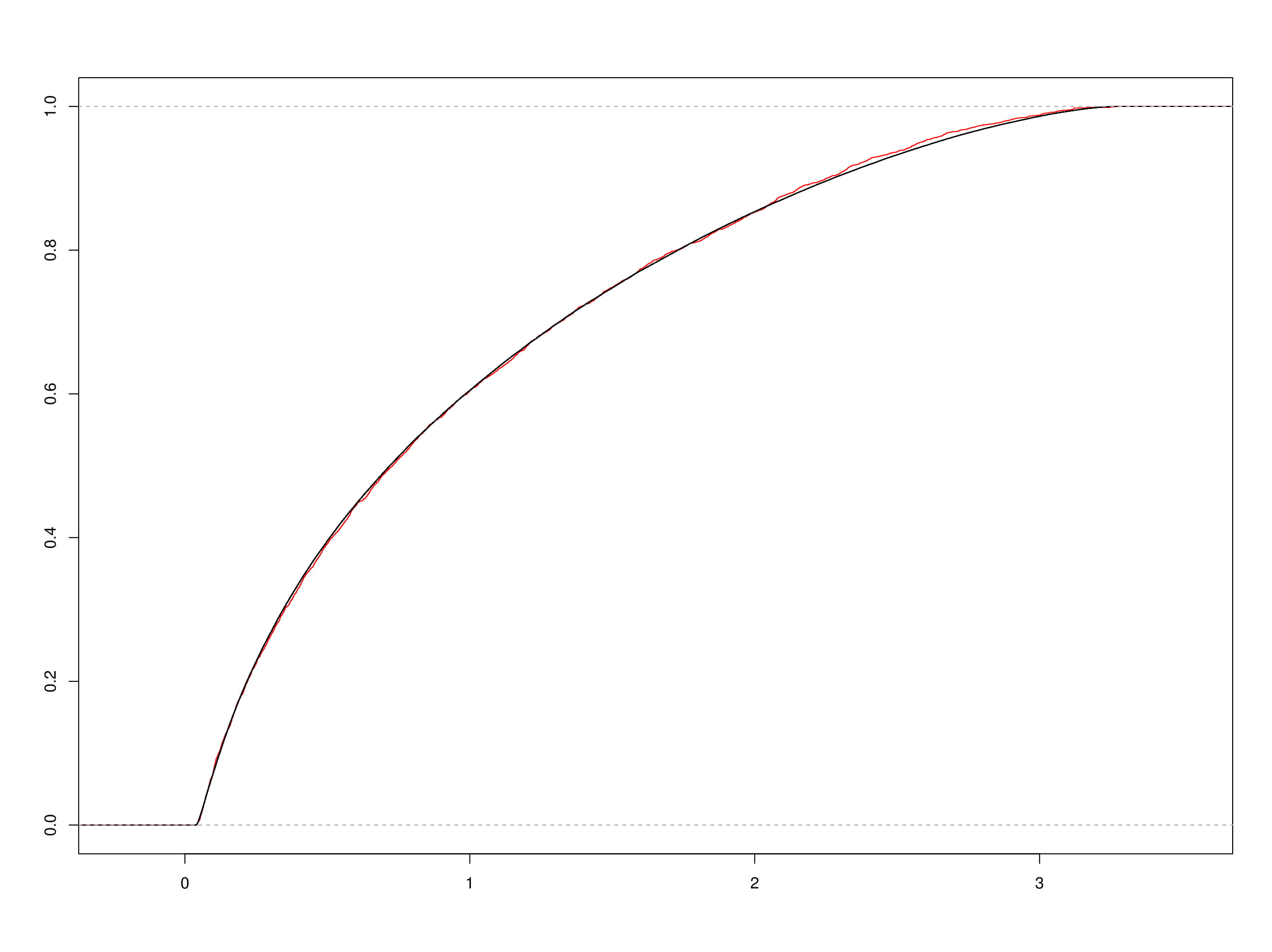} \\
\includegraphics[width = 6.7cm]{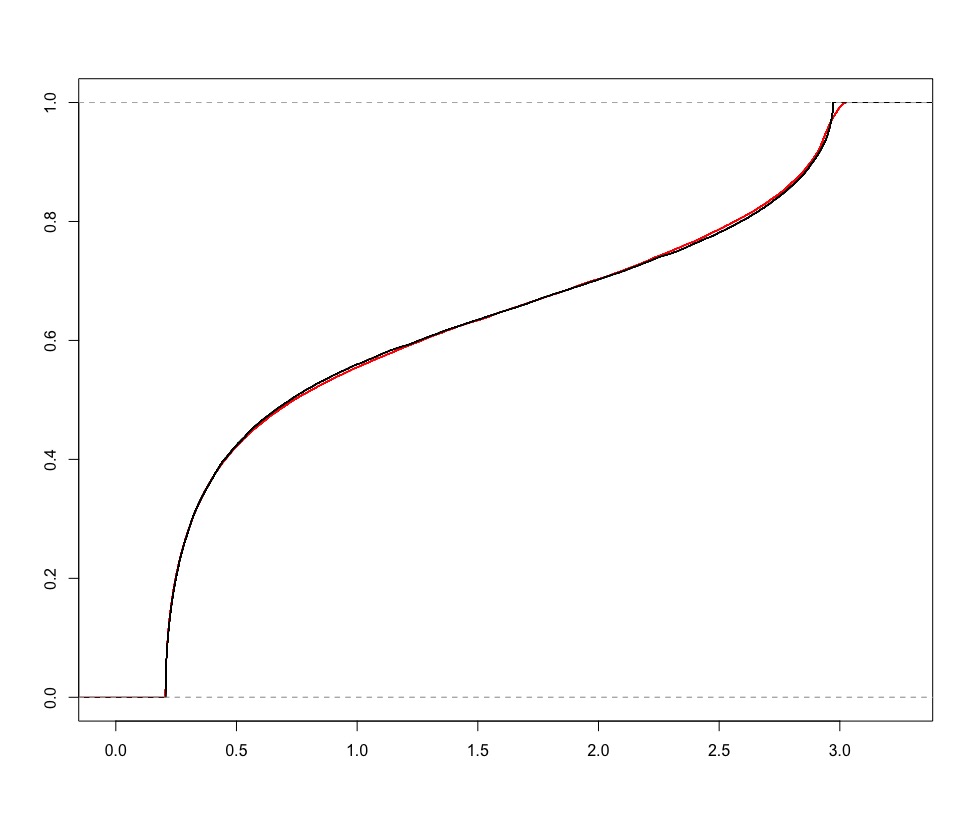}
\includegraphics[width=6.7cm]{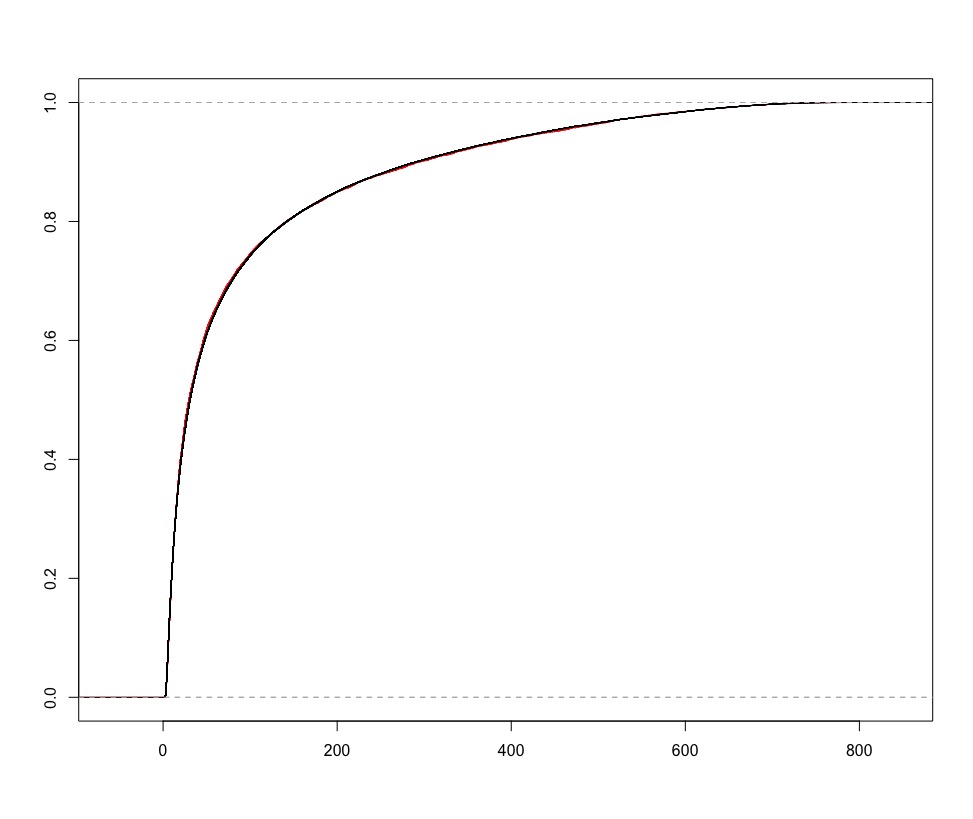}
\caption{Simulated margi6al cumulative distribution functions for the sampled $\hat x_1$ from the conditional samples for the cases of Table~\ref{table:1}. Using the Metropolis-Hastings algorithms of Section~\ref{twopar} (black); using the naive sampler of Section~\ref{naive} (red). Case 1: upper left. Case 2: upper right. Case 3: lower left. Case 4: lower right. } 
		\label{figure1}
\end{figure}

\section{Application to goodness-of-fit testing}
\label{sec6}
As noted in the introduction, a typical use of conditional samples given sufficient statistics is in goodness-of-fit testing.

Consider the null hypothesis $H_0$ that an observation vector $\bX$ comes from a particular distribution indexed by an unknown parameter $\theta$ and such that $\bT = T(\bX)$ is sufficient for~$\theta$. For a test statistic $W(\bX)$ we define the conditional $p$-value by 
\[
  p_{obs}^{W} =P_{H_0} (W(\bX) \ge w_{*} |  \bT=\bt)
\] 
where $w_{*}$ is the observed value of the test statistic and $\bt$ is the observed value of the sufficient statistic. A conditional goodness-of-fit test based on $W$ rejects $H_0$ at significance level $\alpha$ if $ p_{obs}^{W} \le \alpha$. 
Let now $\hat \bx_j$ for $j=1,2,\ldots,k$ be samples from the conditional distribution of $\bX$ given $\bT=\bt$. Then the observed $p$-values are approximated by
\beq
\label{obsp}
    p_{obs}^{W} \approx \frac{1}{k} \sum_{j=1}^k I(W(\bx_j) \ge w_{*}).
\eeq

Consider now data from \cite{dwo}, giving the precipitation from storms in inches at the Jug bridge in Maryland, USA. The observed data are
\begin{align*}&1.01, \ 1.11,\ 1.13,\ 1.15,\ 1.16,\ 1.17,\ 1.2,\ 1.52,\ 1.54,\ 1.54,\ 1.57,\ 1.64,\\
&1.73,\ 1.79,\ 2.09,\ 2.09,\ 2.57,\ 2.75,\ 2.93,\ 3.19,\ 3.54,\ 3.57,\ 5.11,\ 5.62
\end{align*}
comprising the data vector $\bx=(x_1,x_2,\ldots , x_n)$, where $n=24$. 
 The question is whether the gamma or inverse Gaussian distributions fit the data. Using the setup and notation from Section~\ref{twopar} we  calculate the sufficient statistics as 
$$t_1=\sum_{i=1}^n x_i = 52.72, \ \ \ t_2 = \sum_{i=1}^n \log x_i = 15.7815$$
for the gamma distribution and 
$$t_1 = \sum_{i=1}^n x_i = 52.72, \ \ \ t_2=\sum_{i=1}^n \frac1{x_i} = 13.8363$$
for the inverse Gaussian distribution.  
 
Some common test statistics for goodness-of-fit testing are constructed as follows. 
Let  $(x_{(1)}, x_{(2)}, \ldots , x_{(n)})$ be the order statistic of $\bx$. Then define 
the transformed values $z_i = F(x_{(i)} \ ; \  \hat \theta_1, \hat \theta_2)$, where $F(\cdot; \theta_1,\theta_2)$ is the cumulative distribution function of the gamma or inverse Gaussian distributions with parameters $\theta_1, \theta_2$, while  $\hat \theta_1, \hat \theta_2$ are the maximum likelihood estimates which can be found from the corresponding $t_1$ and $t_2$. 

From this setup we can write down the following test statistics: 
\begin{description}
\item[Kolmogorov-Smirnov test \citep{razali}] 
$$D = \max_{1\leq i\leq n}\left( z_i-\frac{i-1}{n}, \frac{i}n-z_i\right).$$
\item[Anderson-Darling test \citep{stephens}] 
$$A^2 = -n-\frac1n\sum_{i=1}^n(2i-1)\left(\ln z_{i}+\ln (1-z_{n-i+1})\right).$$
\item[The Cram\'er-von Mises test \citep{stephens}]
$$\omega^2 = \frac1{12n}+\sum_{i=1}^n\left(z_i-\frac{2j-1}{2n} \right)^2.$$ 
\end{description}
Now let $A^2_*$, $\omega_*$, $D_*$ denote the observed values of the test statistics as calculated from the observed data $\bx$. The approximated conditional $p$-values $p_{obs}^{D}$, $p_{obs}^{A^2}$, $p_{obs}^{\omega^2}$ can now be calculated from (\ref{obsp}) for each  the null hypotheses of gamma distribution and inverse Gaussian distribution, respectively.

We simulated $k=10^5$ samples from the conditional distributions and obtained the results of Table~\ref{goff}.
\begin{table}
	\begin{center}
		\begin{tabular}{||c || c c ||}  	
			\hline
			Test & Inverse Gaussian distribution & Gamma distribution \\ 
			\hline\hline
			$A^2$ & $0.094$ & $0.024$ \\ 
			\hline
			$\omega^2$ & $0.102$ & $0.031$ \\ 
			\hline
			$D$ & $0.217$ & $0.061$ \\ 
			\hline
		\end{tabular}
	\end{center}
	\caption{Conditional $p$-values}
	\label{goff}
\end{table}
The calculated conditional $p$-values indicate that the fit of the inverse Gaussian is marginal, which agrees with the results of \cite{dwo}. Using significance level $\alpha=0.05$, the tests based on $A^2$ and $\omega^2$ suggest that the gamma distribution does not fit the data. 

\section{Concluding remarks}
\label{sec7}
\subsection{Classical conditional Monte Carlo}
The method of the present paper can be seen as a reformulation of the main ideas of the classical concept of conditional Monte Carlo as introduced in the 1950s. The basic idea of conditional Monte Carlo was essentially the introduction of new coordinates.  \cite{tukey} made a point of the ``skullduggery'' related to such  arbitrary new variables which had ``nothing to do with the way our samples were drawn''. This ``trick'' was, however, the successful ingredient of the method, and is basically also the way our method works. The main new coordinate of our approach is represented by a parameter in an artificial statistical model. 

\subsection{Comparison to \cite{LT05}}
As indicated in the Introduction, there are some basic differences between the method of \cite{LT05} and the present approach. Still, the methods share several important ingredients, and we have therefore found it useful to adopt much of the notation from \cite{LT05} in the present paper. With this, both methods end up with the goal of calculating  conditional distributions of certain functions $\chi(\bU,\Theta)$ given related functions  $\tau(\bU,\Theta)=\bt$. While the role of $\Theta$ is apparently very similar in the two methods, there is indeed a difference. As shown in Section~\ref{thetat}, $\Theta$ can in the present approach basically be given any bounded distribution, but not an improper distribution. This is in contrast to  \cite{LT05}, where the distribution of $\Theta$ plays a role more in line with Bayesian and fiducial statistics. The framework and methods of \cite{LT05}, although tailored for the special situation of conditional sampling under sufficiency, in fact also induces interesting algorithms for calculation of Bayesian posterior distributions as well as fiducial sampling. 

\subsection{The roles of $\theta$ and $\chi(\bu,\theta)$}
As we have seen, the parameter $\theta$ will normally have the same dimension as the statistic $T(\bX)$. This ensures that the number of equations to solve for obtaining the $\hat \theta(\bu,\bt)$ is the same as the number of unknowns (see Assumption~1). In the examples of Sections~\ref{411} and \ref{412} we considered a one-dimensional $T(\bX)$, using the ``scaling'' transformation, $\chi(u,\theta)=u/\theta$. In Section~\ref{twopar} we conditioned on a two-dimensional statistic and used the transformation $\chi(u,\theta)=(u/\beta)^\alpha$ which is appropriate for positive variables. This transformation would not be appropriate, however, for models where the observations have support in all of $\RealN$. In this case, the linear transformation $(u-\alpha)/\beta$ could be used instead.  This would for example be a suitable transformation if, in the example of Section~\ref{412}, we conditioned on the average $\bar X$ in addition to the range $\max X_i - \min X_i$.

\subsection{Conditioning on $T(\bX)$ with dimension $k > 2$}
Our initial motivation for the paper came from the conditional sampling given sufficient statistics in two-parameter models like gamma and inverse Gaussian distributions. Still a natural question is, of course,  what to do if we want to condition on $T(\bX)$ with dimension $k > 2$. For the i.i.d.\ case with $X_i$ having support in all of $\RealN$, an obvious choice might be to let $\btheta=(\theta_0,\theta_1,\ldots,\theta_{k-1})$ and 
\beq
\label{poly}
   \chi(u,\btheta) =  \sum_{j=0}^{k-1} \theta_j u^j.
\eeq
If we put $k=2$ in (\ref{poly}), then this is in fact equivalent to the transformation $(u-\alpha)/\beta$ as suggested above. 

In the i.i.d.\ case with positive $X_i$, a general suggestion might be to use 
\beq
\label{nykji}
   \chi(u,\btheta) = \exp \left\{ \sum_{j=0}^{k-1} \theta_j u^j \right\}.
\eeq
For $k=2$ this transformation is in fact equivalent to the transformation used for the two-parameter exponential families of positive variables in Section~\ref{twopar}, since
\beq
\nonumber
    \left( \frac{u}{\beta}\right)^\alpha = \exp \{-\alpha \log \beta + \alpha \log u \}.
\eeq      
It follows from this that, in the gamma and inverse Gaussian cases treated in Section~\ref{twopar}, we could as well have used the transformation (\ref{nykji}) with $k=2$, and still obtained unique solutions for $\hat \btheta$.  In general, however, there might be several solutions for $\btheta$ in the equations $\tau(\bu,\btheta)=\bt$. There would therefore be a need for the possibility of relaxing Assumption 1 to allow more than one solution of the equation $\tau(\bu,\theta)=\bt$. We sketch an approach below. 

\subsection{Multiple solutions of the equation $\tau(\bu,\theta)=t$}
\label{multsol}
In general it might be difficult or impossible to find a suitable function $\chi(\bu,\theta)$ such that Assumption~1 holds. In practice, there may be a finite number of solutions, where the number may also depend on the values of $(\bu,\bt)$. Define then
\[
  \Gamma(\bu,\bt) = \{ \hat \theta : \tau(\bu,\hat \theta)=\bt\}.
\]
An extension of the arguments leading to (\ref{huttetu}), taking into account the multiplicity of the roots of the equation $\tau(\bu,\theta)=\bt$, then gives the following expression for the joint density of $(\bU,\Theta, \tau(\bU,\Theta))$,
\begin{equation}
\label{mult}
 h(\bu,\hat \theta, \bt) =  f(\bu|\hat \theta)
  \left| \frac{\pi(\theta)}{\det \partial_\theta \tau(\bu,\theta)} \right|_{\theta=\hat \theta}
 \end{equation}
 for $\bu,\bt$ as before, and  $\hat \theta \in \Gamma(\bu,\bt)$.  
 A similar expression was obtained in \cite{doksum}.
 
 The formula (\ref{efi}) for conditional expectations now becomes
 \[
   \E[\phi(\bX)|\bT=\bt] 
      = \frac{\int \sum_{\hat \theta \in \Gamma(\bu,\bt)}\phi(\chi(\bu,\hat \theta)) h(\bu,\hat\theta,\bt) d\bu }
      {\int \sum_{\hat \theta \in \Gamma(\bu,\bt)} h(\bu,\hat \theta, \bt) d\bu } ,
      \]
while the Metropolis-Hastings method of Section~\ref{213} may proceed as follows. First, propose the $\bu$ in the same way as in Section~\ref{213}, and then calculate the roots $\hat \theta \in \Gamma(\bu,\bt)$. One of these roots, say $\hat \theta'$, is then chosen at random according to the conditional distribution of $\hat \theta$ given $\bu$ and $\bt$, as found from (\ref{mult}). A properly modified version of the criterion  (\ref{mcmc}) is then used for acceptance or non-acceptance, using $h(\bu',\hat \theta',\bt)$ instead of  $\tilde h(\bu,\bt)$.

\subsection{Using the pivot $\tau(\bU,\theta)$ in statistical inference}
As noted in Section~\ref{sec2}, the random vector $\chi(\bU,\theta)$ and hence also $\tau(\bU,\theta)$  are pivots in the constructed artificial statistical model. In order to study their possible properties in a statistical inference setting, recall that for the gamma distribution case of Section~\ref{421}, we used $f_X(x)=e^{-x}$ and the transformation  (\ref{kjiexp}). In this case, (\ref{futh}) is in fact the joint density of $n$ i.i.d.\ Weibull-distributed random variables  with shape parameter $\alpha$ and scale parameter $\beta$. A curious biproduct of our method is therefore the construction of exact confidence sets for the pair $(\alpha,\beta)$ from observed i.i.d.\ Weibull-distributed data $\bu = (u_1,u_2,\ldots,u_n)$. The basis of the confidence sets would then be to sample vectors $\bx$ from the unit exponential distribution, calculate $t_1=\sum x_i$ and $t_2=\sum \log x_i$  and solve (\ref{gamma1})-(\ref{gamma2}) for $\alpha$ and $\beta$ with $\bu$ fixed at the observed Weibull-data. The resulting pairs $(\hat \alpha,\hat \beta)$ would then have a joint distribution corresponding to a two-dimensional confidence distribution for $(\alpha,\beta)$.  (This is in some sense exactly the opposite of what we are doing in Section~\ref{421}, where $t_1$ and $t_2$ are fixed, and we sample the $u_i$). For exact inference in Weibull models based on the maximum likelihood estimators for $(\alpha,\beta)$ we refer to \cite{weibull}.

\bibliography{mcbib}

\bibliographystyle{chicago}

\newpage
\section*{Appendix}
\begin{lemma}\label{unique1}
Let $n\in \mathbb{N}$ and $u_1,u_2,\ldots ,u_n\in \mathbb{R}^+$, and let for some $v_1,v_2,\ldots ,v_n\in \mathbb{R}^+$,
$$\sum_{i=1}^n v_i=t_1,$$
$$\sum_{i=1}^n \ln v_i=t_2.$$
Then the system of equations
\begin{equation}\nonumber
\left\{ 
\begin{array}{lcc}
\sum_{i=1}^n   \left( \frac{u_i}{\beta} \right)^\alpha & =& t_1 \\
\sum_{i=1}^n \ln  \left( \frac{u_i}{\beta} \right)^\alpha  & =& t_2, \end{array}
\right.
\end{equation}
has a unique solution for $\alpha, \beta \in \mathbb{R}^+$.
	\end{lemma}

\begin{proof}
	We can transform the system into
	\[
	\left\{ 
	\begin{array}{lcc}\sum_{i=1}^n \left(\frac{u_i}{\beta}\right)^\alpha &=&t_1 \\ \\ \frac{\sum_{i=1}^n u_i^\alpha}{\left( \prod_{i=1}^n u_i^\alpha \right)^{1/n}}&=& \frac{t_1}{\exp(t_2/n)} \end{array} \right.
	\]
	If the  function 
	$$p(\alpha)=\frac{\sum_{i=1}^n u_i^\alpha}{\left( \prod_{i=1}^n u_i^\alpha \right)^{1/n}}$$
	is monotone, then there is a unique solution. The derivative is 
	\begin{align*}
	p'(\alpha)&=\left(\sum_{i=1}^n\left( \frac{ u_i}{\left( \prod_{i=1}^n u_i \right)^{1/n}}\right)^\alpha\right)' \\
	&=\sum_{i=1}^n \left( \frac{ u_i}{\left( \prod_{i=1}^n u_i \right)^{1/n}}\right)^\alpha \ln \frac{ u_i}{\left( \prod_{i=1}^n u_i \right)^{1/n}}.
	\end{align*}
	We note that $\lim_{\alpha\to0^+}p'(\alpha)=0$. The second derivative is 
	$$p''(\alpha)=\sum_{i=1}^n \left( \frac{ u_i}{\left( \prod_{i=1}^n u_i \right)^{1/n}}\right)^\alpha \ln^2 \frac{ u_i}{\left( \prod_{i=1}^n u_i \right)^{1/n}} \geq 0.$$
	Since the second derivative is positive, the first derivative is increasing. Hence we can conclude that the first derivative is always positive and $p$ is increasing. The solution exists, since 
	$$\lim_{\alpha\to 0^+}p(\alpha)=n$$
	and
	$$\frac{t_1}{\exp(t_2/n)}=\frac{\sum_{i=1}^n v_i}{\left( \prod_{i=1}^n v_i\right)^{1/n}} \geq n.$$
	The last inequality holds because the arithmetic mean is always larger than or equal to the geometric mean.
	\end{proof}

\begin{lemma}\label{unique2}
	Let $n\in \mathbb{N}$ and $u_1,u_2,\ldots ,u_n\in \mathbb{R}^+$, and let for some $v_1,v_2,\ldots ,v_n\in \mathbb{R}^+$,
	$$\sum_{i=1}^n v_i=t_1,$$
	$$\sum_{i=1}^n v_i^{-1}=t_2.$$
	Then the system of equations
	\begin{equation}\nonumber
	\left\{  \begin{array}{lcc} \sum_{i=1}^n  \left( \frac{u_i}{\beta} \right)^\alpha & =& t_1\\
\sum_{i=1}^n  \left( \frac{u_i}{\beta} \right)^{-\alpha}  = t_2,\end{array}  \right.
	\end{equation}
	has a unique solution for $\alpha, \beta \in \mathbb{R}^+$.
\end{lemma}
\begin{proof}
	We can transform the system into
	\[
	\left\{ \begin{array}{lcc}  \sum_{j=1}^n u_j^\alpha \sum_{i=1}^n u_i^{-\alpha} &= &t_1t_2 \\ \\
	\sum_{i=1}^n\left( \frac{u_i}{\beta}\right)^{-\alpha}&=& t_2 
	\end{array} \right.
	\]
	If the function 
	$$p(\alpha)=\sum_{j=1}^nu_j^\alpha \sum_{i=1}^n u_i^{-\alpha}$$
	is monotone, then there is a unique solution for $\alpha$.   In order to prove the monotonicity, let  $y_{ij}=\frac{u_j}{u_i}$, where $i,j=1,2,\ldots ,n$, $i\neq j$. The derivative is
	\begin{align}
	p'(\alpha)&= \left(\sum_{j=1}^n\sum_{i=1}^n \left( \frac{u_j}{u_i}\right)^\alpha\right)'\nonumber \\
	&=\left(\sum_{j=1}^n\sum_{i=1}^n y_{ij}^\alpha\right)'\nonumber \\
	&=\sum_{j=1}^n\sum_{i=1}^n y_{ij}^\alpha\ln y_{ij}\nonumber \\
	&=\sum_{i<j}\ln y_{ij}\left(y_{ij}^\alpha-y_{ij}^{-\alpha}\right).	\label{mono}
	\end{align}
	Now, if $y_{ij}> 1$, then $\ln y_{ij}>0$ and $y_{ij}^\alpha >y_{ij}^{-\alpha}$, which means that 
	$$\ln y_{ij}\left(y_{ij}^\alpha-y_{ij}^{-\alpha}\right)>0.$$
	If $y_{ij}< 1$, then $\ln y_{ij}<0$ and $y_{ij}^\alpha < y_{ij}^{-\alpha}$, which means that 
	$$\ln y_{ij}\left(y_{ij}^\alpha-y_{ij}^{-\alpha}\right)>0.$$
	Hence, we can conclude that \eqref{mono} is positive and the function $p$ is increasing. Since 
	$$\lim_{\alpha\to 0^+} p(\alpha)=n^2$$
	and
	$$t_1t_2 = \sum_{i=1}^n v_i\sum_{i=1}^n v_i^{-1}  \geq n^2$$
	the solution always exists.
\end{proof}

 \end{document}